\documentclass[final, runningheads]{llncs}
\usepackage{microtype,verbatim,dsfont}
\usepackage{hyperref,xspace,xcolor}
\usepackage{todonotes}
\usepackage{graphicx}
\usepackage[subrefformat=parens,labelfont=default]{subcaption}
\usepackage{enumitem}

\usepackage{amsmath,amssymb,amsthm}
\usepackage{enumitem}

\graphicspath{{figures/}}

\renewcommand{\orcidID}[1]{\href{https://orcid.org/#1}{\includegraphics[scale=.03]{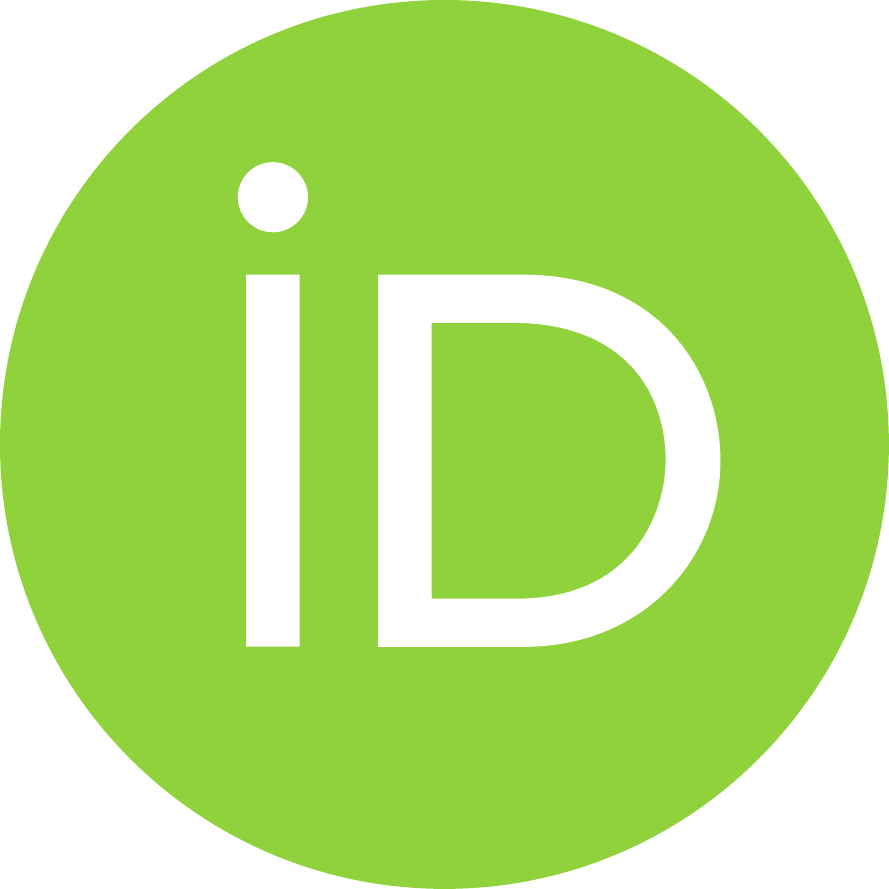}}}

\newcommand{\crsv}[1]{\textsc{CRSV}($#1$)}
\newcommand{\crs}[1]{\textsc{CRS}($#1$)}

\renewcommand{\emph}[1]{\textcolor{blue!80!black}{\it #1}\xspace}

\usepackage[utf8]{inputenc}

\newtheorem*{problem*}{Problem}
\newtheorem*{theorem*}{Theorem}
\newtheorem{iclaim}{Claim}
\newtheorem{observation}{Observation}

\begin{document}

\title{An FPT Algorithm for Bipartite Vertex Splitting}
\author{
Reyan~Ahmed\inst{1} \and
Stephen~Kobourov\inst{2}\orcidID{0000-0002-0477-2724} \and 
Myroslav~Kryven\inst{2}
} 

\authorrunning{Ahmed, Kobourov, Kryven}
%
\institute{Colgate University\\
\email{abureyanahmed@arizona.edu} \and
University of Arizona\\
\email{\{kobourov,myroslav\}@arizona.edu}}

\maketitle              
\begin{abstract}
Bipartite graphs model the relationship between two disjoint sets of objects.
They have a wide range of applications and are often visualized as 2-layered drawings,
where each set of objects is visualized as vertices (points) on one of two parallel horizontal lines and the relationships
are represented by (usually straight-line) edges between the corresponding vertices.
One of the common objectives in such drawings is to minimize the number of crossings. This, in general,
is an NP-hard problem and may still result in drawings with so many crossings that they affect the readability of the drawing.
We consider a recent approach to remove crossings in such visualizations by splitting vertices,
where the goal is to find the minimum number of vertices to be split to obtain a planar drawing. We show that determining whether a planar 2-layered drawing exists after splitting at most $k$ vertices is fixed parameter tractable in $k$.

\end{abstract}
\section{Introduction}

Bipartite graphs are used in many applications to study complex systems and their dynamics~\cite{pavlopoulos2018bipartite}. We can visualize a bipartite graph $G=(T\cup B, E)$ as a 2-layered drawing where vertices in $T$ are placed (at integer coordinates) along the horizontal line defined by $y=1$ and vertices in $B$ along the line below (at integer coordinates) defined by $y=0$.

A common optimization goal in graph drawing is to minimize the number of crossings. Deciding whether a planar 2-layered drawing exists for a given graph can be done in linear time, although most graphs, including sparse ones such as cycles and binary trees, do not admit planar 2-layered drawings~\cite{eades1986edge}. The problem of minimizing the number of crossings in 2-layered layouts is NP-hard, even if the maximum degree of the graph is at most four~\cite{munoz2001one}, or if the permutation of vertices is fixed on one of the layers~\cite{eades1986edge}.
The latter variant of the problem is known as One-Sided Crossing Minimization (OSCM).  The minimum number of crossings in a 2-layered drawing can be approximated within a factor of $1.47$ and $1.3 +12/(\delta-4)$, where $\delta$ is the minimum degree, given that $\delta > 4$~\cite{nagamochi2003improved}. 
Dujmovi\'c et al.~\cite{DujmovicW04} gave a fixed-parameter tractable (FPT) algorithm with runtime $O(1.62^k\cdot n^2)$, which was later improved to $O(1.4656^k + kn^2)$~\cite{dujmovic2008fixed}. 
Fernau et al.~\cite{fernau2010ranking} 
reduced this problem to weighted FAST (feedback arc sets in tournaments) obtaining a subexponential time algorithm that runs in time $2^{O(\sqrt{k}\log{k})} +n^{O(1)}$. Finally Kobayashi and Tamaki~\cite{kobayashi2012fast} gave a straightforward dynamic programming algorithm on an interval graph associated with each OSCM instance with runtime~$2^{O(\sqrt{k}\log{k})} +n^{O(1)}$. They also showed that the exponent $O(\sqrt{k})$ in their bound is asymptotically optimal under the Exponential Time Hypothesis (ETH)~\cite{impagliazzo2001complexity}, a well-known complexity assumption which states that, for each $k \geq 3$, there is a positive constant $c_k$ such that $k$-SAT cannot be solved in $O(2^{c_k n})$ time where $n$ is the number of variables. 


Minimizing the number of crossings in 2-layer drawings may still result in visually complex drawings from a practical point of view~\cite{jm-2scmpeha-97}. Hence, we study vertex splitting~\cite{em-vtl-95,Liebers-survey01,ekkllm-pstg-18,Knauer2016} which aims to construct planar drawings, and thus, avoid crossings altogether. In the \emph{split} operation for a vertex $u$ we delete $u$ from $G$, add two new copies $u_1$ and $u_2$, and distribute the edges originally incident to $u$ between the two new vertices $u_1$ and $u_2$. 
There are two main variations of the objective in vertex splitting: minimizing the number of split operations (or \emph{splits}) and minimizing the number of \emph{split vertices} (each vertex can be split arbitrary many times) to obtain a planar drawing of $G$. 
Minimizing the number of splits is NP-hard even for cubic graphs~\cite{faria_splitting_2001nourl}. 
Nickel et al.~\cite{nickel-eurocg22} extend the investigation of the problem and its complexity from abstract graphs to drawings of graphs where splits are performed on an underlying drawing.

Vertex splitting in bipartite graphs with 2-layered drawings has not received much attention~\cite{dagrep2l}.
In several applications, such as visualizing graphs defined on anatomical structures and cell types in the human body \cite{website}, the two vertex sets of $G$ play different roles and vertex splitting is
allowed only on one side of the layout. This
has motivated the interest in splitting the vertices in only one vertex
partition of the bipartite graph.
It has been shown that minimizing splits in this setting is NP-hard for an arbitrary bipartite graph~\cite{DBLP:journals/tcad/ChaudharyCHNRW07}. 

The other variant --  minimizing the number of split vertices -- has been recently considered and was shown to be NP-hard~\cite{ahmed-split-vertices}. 
On the positive side, we show that the problem is FPT parameterized by the natural parameter, that is, the number of split vertices.

\begin{problem*}[Crossing Removal with $\boldsymbol{k}$ Split Vertices  -- \crsv{k}]
Let $G = (T \cup B, E)$ be a bipartite graph. Decide whether there is a planar 2-layer drawing of $G$ after splitting  at most $k$ vertices of $B$.
\end{problem*}

In the next section we prove the following theorem.

\begin{theorem}
\label{thm:crsvk-fpt}
Given a bipartite graph $G = (T \cup B, E)$, the \crsv{k} problem can be decided in time $2^{O(k^6)}\cdot m$, where $m$ is the number of edges of $G$.
\end{theorem}


We prove Theorem~\ref{thm:crsvk-fpt} using \emph{kernelization}, one of the standard techniques for designing FPT algorithms. The goal of kernelization is to reduce the input instance to its computationally hard part on which a slower exact algorithm can be applied. If the size of the reduced instance is bounded by a function of the parameter, the problem can be solved by brute force on the reduced instance yielding FPT runtime.  Our reduction consists of two parts. 

In the first part we identify and remove vertices that necessarily belong to the solution (Step~\ref{step1} below) and remove redundant vertices of the input graph $G$; Step~\ref{step2} below. Then we show that there is a solution for the reduced graph $G''_1$ if and only if there is a solution for the original graph $G$; see Claim~\ref{cl:g-g1}. 
Then we prove two structural properties about the degrees of the vertices of $G''_1$; see Lemmas~\ref{lem:obs-max-deg} and~\ref{lem:num-high-deg-vtc}. These two properties allow us to bound the size of the ``essential'' part (called the \emph{core}) of the reduced graph $G''_1$; see Lemma~\ref{lem:core-size}. 
%

In the second part of the reduction we remove more redundant vertices of $G''_1$ and identify and remove the vertices that necessarily belong to the solution. Then we show that the resulting reduced graph $G'_2$ has size bounded by a polynomial function of the parameter; see Lemma~\ref{lem:g2-size}. Finally, we show that there is a solution for $G'_2$ if and only if there is a solution for $G''_1$; see Claim~\ref{cl:g1-g2}. The proof is concluded by applying an exact algorithm to the graph~$G'_2$.

\section{Proof of Theorem~\ref{thm:crsvk-fpt}}
Let $G=(T\cup B, E)$ be a bipartite graph and $k$ be the number of vertices that we are allowed to split.

\paragraph{First reduction rule:}
Before we describe our first reduction rule, we make a useful observation.

\begin{observation}
\label{obs:must-split}
If a vertex $v \in B$ has at least three neighbours of degree at least two, it must be split in any planar 2-layered drawing of $G$; see 
Figure~\ref{fig:a-vertex-in-B}.
\end{observation}
Let $B_\text{tr}$ be the set of such vertices of degree 3 or more in $B$ (as described in Observation~\ref{obs:must-split}).
The first reduction rule consists of two steps described below. 
\begin{enumerate}
\item We initialize our solution set $\mathcal{S}$ with the vertices in $B_\text{tr}$, that is, $\mathcal{S}:=B_\text{tr}$ and remove them from the graph $G$. Let the resulting graph be $G'_1 = (T'_1\cup B'_1, E'_1)$ and $k'_1 = k - |B_\text{tr}|$; note that $T'_1 = T_1$. \label{step1}

\item  Let $T_s\subset T'_1$ be the set of vertices $v$ such that deg$(v)= 1$ and deg$(u) \ge 3$, where $u$ is the unique neighbor of $v$ in $G'_1$. Similarly, let $B_s\subset B'_1$ be the set of vertices $v$ such that deg$(v)= 1$ and deg$(u) \ge 3$, where $u$ is the unique neighbor of $v$ in $G'_1$. We remove the vertices $T_s$ and $B_s$ from the graph $G'_1$. 
Let the resulting graph be $G''_1 = (T''_1\cup B''_1, E''_1)$.
\label{step2} 
\end{enumerate}

Let us now show the following.
\begin{iclaim}
\label{cl:g-g1}
The graph $G$ is a \textsc{Yes} instance for \crsv{k} if and only if $G''_1$ is a \textsc{Yes} instance for \crsv{k_1'}. 
\end{iclaim}
\begin{proof}
We first argue the ``only if" direction: consider a planar 2-layered drawing of $G$ with at most $k$ vertices split.
According to Observation~\ref{obs:must-split}, each vertex in $B_\text{tr}$ is split, moreover, none of the vertices in $B_s$ are split because each of them has degree one. Therefore, there are at most $k-|B_\text{tr}|$ vertices in $B \setminus (B_\text{tr} \bigcup B_s)$ that are split. Because $B''_1 = B \setminus (B_\text{tr} \bigcup B_s)$ and $k'_1 = k-|B_\text{tr}|$ there exists a planar 2-layered drawing of $G''_1$ with at most $k'_1$ vertices split.

For the ``if" direction, consider a planar 2-layered drawing of $G''_1$ with at most $k'_1$ vertices split. Note that after applying Step~\ref{step1} and Step~\ref{step2} the vertices in $B''_1$ have degree at most two. Thus for each vertex $v\in B_\text{tr}$ we can reinsert its split copies $v_1, v_2, \dots, v_{\text{deg}{(v)}}$ (each reinserted vertex has degree one) without crossings; see Figure~\ref{fig:g1-g-B_3}. 
For the same reason we can reinsert the vertices in $B_s$ of degree one removed at Step~\ref{step2}; see Figure~\ref{fig:singletones-in-T}. To see that we can reinsert each vertex $u\in T_s$ of degree one removed at Step~\ref{step2} observe that we always connect it to a vertex $v\in B''_1$ of degree at least two, therefore, in any planar 2-layered drawing of $G''_1$ there is always a \emph{safe wedge} formed by two edges $vv_1$ and $vv_2$ where we can fit in the edge $vu$ without causing any crossings; see Figure~\ref{fig:singletones-in-T}.
\end{proof}
\begin{figure}[tb]
  \begin{subfigure}[t]{0.28\textwidth}
    \centering
    \includegraphics[page=1]{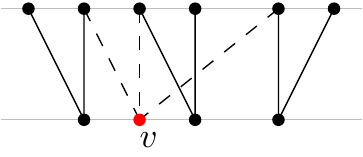}
    \caption{$v\in B_\text{tr}$}
    \label{fig:a-vertex-in-B}
  \end{subfigure}
  \hfill
  \begin{subfigure}[t]{0.35\textwidth}
    \centering
    \includegraphics[page=2]{g1-g}
    \caption{a drawing of~$G''_1$}
    \label{fig:2-layered-drawing-G''_1}
  \end{subfigure}
  \hfill
  \begin{subfigure}[t]{0.35\textwidth}
    \centering
    \includegraphics[page=3]{g1-g}
    \caption{reinserting split copies of $v$}
    \label{fig:reinserting-split-copies}
  \end{subfigure}
  \caption{
   Reinserting split copies $v_1, v_2, \dots, v_{\text{deg}(v)}$ of $v\in B_\text{tr}$ into a planar 2-layered drawing of $G''_1$ to get a planar  2-layered drawing of $G$. }
    \label{fig:g1-g-B_3}
\end{figure}
  \begin{figure}[t]{}
    \centering
    \includegraphics{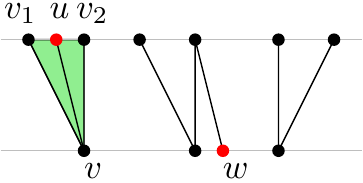}
    \caption{
    %
    %
    Reinserting $w\in B_s$ and $u\in T_s$ into a 2-layered drawing of $G''_1$ to obtain a 2-layered drawing of $G$. A safe wedge $v_1vv_2$ is filled green.
    }
    \label{fig:singletones-in-T}
  \end{figure}

Now we state two observations about the degrees of the vertices of the graph~$G''_1$.

\begin{lemma}
\label{lem:obs-max-deg}
For each vertex $v\in T''_1$ it holds that deg$(v)\le k'_1+2$ if there exists a planar 2-layered drawing of $G''_1$ with at most $k'_1$ split vertices.
\end{lemma}
\begin{proof}
Consider for contradiction that there is a vertex $v\in T''_1$ that has deg$(v) = k'_1+3$; see Figure~\ref{fig:max-deg}. According to Step~\ref{step2} $v$ does not have any neighbors of degree one in $B''_1$, therefore, to obtain a planar 2-layered drawing of $G''_1$ all but two neighbors of $v$ must be split, that is, $k'_1+1$ vertices must be split; contradiction.
\end{proof} 
\begin{figure}
\centering
\includegraphics{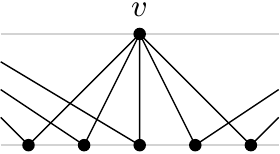}
\caption{All but two neighbors of $v$ must be split to obtain a planar 2-layered drawing of $G''_1$ because each neighbour of $v$  has degree at least two.}
\label{fig:max-deg}
\end{figure}

To make our second observation let $T''_{1,~{\text{deg}(v)\ge 3}}$ be the set of all the vertices of degree at least three in $T''_1$. 
\begin{lemma}
\label{lem:num-high-deg-vtc}
It holds that $\big\vert T''_{1,~{\text{deg}(v)\ge 3}}\big\vert \le 2k'_1$ if there exists a planar 2-layered drawing of $G''_1$ with at most $k'_1$ split vertices.
\end{lemma}
\begin{proof}
Observe that according to Step~\ref{step2} no vertex in $T''_{1,~{\text{deg}(v)\ge 3}}$ has any neighbors of degree one in $B''_1$. This implies that for each vertex $v$ in $T''_{1,~{\text{deg}(v)\ge 3}}$ at least one of its neighbors $u\in B''_1$ must be split to obtain a planar 2-layered drawing of $G''_1$; see Figure~\ref{fig:high-deg}. But the degree of $u$ is at most two, therefore, splitting $u$ can resolve crossings for at most two vertices $v_1, v_2 \in T''_{1,~{\text{deg}(v)\ge 3}}$. Thus,  if $|T''_{1,~{\text{deg}(v)\ge 3}}| > 2k'_1$, more than $k'_1$ vertices in $B''_1$ must be split to obtain a planar 2-layered drawing of $G''_1$; contradiction. 
\end{proof} 

\begin{figure}[tb]
  \begin{subfigure}[t]{0.49\textwidth}
    \centering
    \includegraphics[page=1]{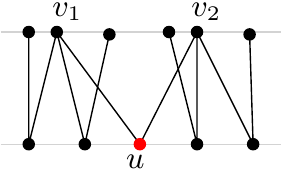}
  \end{subfigure}
  \hfill
  \begin{subfigure}[t]{0.49\textwidth}
    \centering
    \includegraphics[page=2]{high-deg}
  \end{subfigure}
  \caption{For each vertex in $v \in T''_{1,~{\text{deg}(v)\ge 3}}$ at least one of its neighbors $u\in B''_1$ must be split to obtain a planar 2-layered drawing of $G''_1$ because each of the neighbours of $v$ has degree at least two. Splitting $u$ can resolve crossings for at most two vertices $v_1, v_2 \in T''_{1,~{\text{deg}(v)\ge 3}}$ because it has degree at most two.}
    \label{fig:high-deg}
\end{figure}

For a subset of vertices $W$ let $N(W)$ denote the set of neighbors of $W$. From Lemma~\ref{lem:obs-max-deg} and~\ref{lem:num-high-deg-vtc} we obtain the following.
\begin{lemma}
\label{lem:core-size}
The graph induced by the vertices $T''_{1,~{\text{deg}(v)\ge 3}}\bigcup N(T''_{1,~{\text{deg}(v)\ge 3}})$ has at most $2k'_1(k'_1+2)$ vertices if there exists a planar 2-layered drawing of $G''_1$ with at most $k'_1$ split vertices.
\end{lemma}
Let $C = T''_{1,~{\text{deg}(v)\ge 3}}\bigcup N(T''_{1,~{\text{deg}(v)\ge 3}})$ and call the graph induced by the vertices in $C$ the \emph{core} of $G''_1$.
Now we can proceed to the second reduction rule.

\paragraph{Second reduction rule:}
Observe that all the vertices in $(B''_1 \bigcup T''_1) \setminus C$ have degree at most two, and therefore, induce paths or cycles in  $G''_1$.
Since the cycles are not connected to the core in $G''_1$ (because their vertices have degree at most two in $G''_1$) we can remove them and handle separately. We need to account for one split vertex per each such cycle. Let $\mathcal{E}$ be the set of these cycles and let $k'_2 = k'_1 - |\mathcal{E}|$. In addition, let $Z$ be the set of vertices that we split in these cycles,  $\mathcal{S}:= \mathcal{S} \bigcup Z$.

Let $\mathcal{P}$ be the set of paths induced in $G''_1$ by the vertices in $(B''_1 \bigcup T''_1) \setminus C$ of length at least $2k'_2+5$. We reduce $G''_1$ to $G'_2 = (T'_2\cup B'_2, E'_2)$ by \emph{shortening}  each path $p\in\mathcal{P}$ (that is, iteratively removing one of the middle vertices of $p$ from $T''_1$ and identifying its two neighbours in  $B''_1$) until $p$ has at most $2k'_2+5$ vertices. Because during shortening step the length of $p$ decreases by two, after the shortening process $p$ will still have at least $2k'_2+3$ vertices.

\begin{iclaim}
\label{cl:g1-g2}
The graph $G''_1$ is a \textsc{Yes} instance for \crsv{k'_1} if and only if $G'_2$ is a \textsc{Yes} instance for \crsv{k'_2}. 
\end{iclaim}
\begin{proof}
In one direction the claim is obvious, because shortening paths in a planar 2-layered drawing of the graph $G''_1$ does not cause any crossings.

For the other direction, consider a planar 2-layered drawing of the graph $G'_2$. To obtain from it a planar 2-layered drawing of the graph $G''_1$ we need to: (1) reinsert each of the cycles in $\mathcal{E}$ that we have removed from $G''_1$ to obtain $G'_2$, and (2) reinsert back the missing parts of the paths of $\mathcal{P}$, which are made up of the vertices from $(B''_1 \bigcup T''_1) \setminus C$. Because the cycles in $\mathcal{E}$ are disconnected from $G''_1$ we can reinsert them anywhere in the drawing wherever there is space with one split vertex in $Z$. 

Let us now argue why we can reinsert the missing vertices from $(B''_1 \bigcup T''_1) \setminus C$ into the paths in $\mathcal{P}$; we will refer to Figure~\ref{fig:path-insert} for illustration.
Because for any such path $p \in \mathcal{P}$ the length of $p$ is at least $2k'_2+3$ there must be at least one vertex $v$ in $B'_2$ that was not split in a planar 2-layered drawing of $G'_2$ (see Figure~\ref{fig:path-insert-1}), as otherwise a planar 2-layered drawing of $G'_2$ cannot be constructed with at most $k'_2$ splits. Therefore, there must be a safe wedge formed by the unsplit vertex $v$ and the two edges of the path $p$ incident to $v$ providing space to reinsert the missing vertices without causing any crossings; see Figure~\ref{fig:path-insert-2}.
\end{proof}

\begin{figure}[tb]
  \begin{subfigure}[t]{0.55\textwidth}
    \centering
    \includegraphics[page=1]{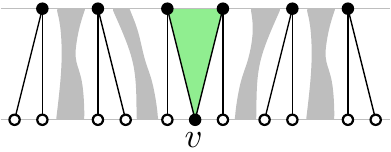}
    \caption{$p$ (black) in $G'_2$, safe wedge (green) }
    \label{fig:path-insert-1}
  \end{subfigure}
  \hfill
  \begin{subfigure}[t]{0.44\textwidth}
    \centering
    \includegraphics[page=2]{path-insert-2}
    \caption{reinserting the missing part of $p$}
    \label{fig:path-insert-2}
  \end{subfigure}
   \caption{Reinserting the missing part of the path $p\in\mathcal{P}$ into a planar 2-layered drawing of $G'_2$ to get a planar  2-layered drawing of $G''_1$.}
    \label{fig:path-insert}
\end{figure}

\begin{lemma}
\label{lem:g2-size}
The graph $G'_2$ has at most $O(k^6)$ vertices. 
\end{lemma}
\begin{proof}
According to Lemma~\ref{lem:core-size} the core $C$ has at most  $2k'_1(k'_1+2)$ vertices and
according to Lemma~\ref{lem:obs-max-deg} the highest degree of each vertex in $C$ is at most $k'_1+2$. Therefore, there can be at most ${{2k'_1(k'_1+2)}\choose{2}}(k'_1+2)$ many paths induced by the vertices in $(B'_2 \bigcup T'_2) \setminus C$.
Moreover, after applying the second reduction rule each such path has at most $2k'_2+5$ vertices.
Thus the total number of vertices in $G'_2$ is at most ${{2k'_1(k'_1+2)}\choose{2}}(k'_1+2) (2k'_2+5) \in O(k^6)$.
\end{proof}

Finally we decide \crsv{k'_2} for $G'_2$ by brute force. More precisely, we check all subsets $X$ of $B'_2$ such that $|X| \le k'_2 \le k$. For each vertex $v$ in $X$ we check all ways to partition its incident edges (at most $k+2$) into non-empty subsets, this represents splitting of $v$. The number of such partitions is bounded by the Bell number of order $k+2$, which in turn is bounded by $(k+2)!$.
Then we run a linear time algorithm to check whether a planar 2-layered drawing of the resulting graph exists. 
This can be done in time $2^{O(k^6)}(k!)^{O(k)} \cdot m \subset 2^{O(k^6)}\cdot m$, where $m$ is the number of edges of $G$.
If $G'_2$ is a \textsc{yes} instance for \crsv{k'_2} with the subset of split vertices $X$, we update our solution set $\mathcal{S}:=\mathcal{S}\bigcup X$ and return it. It is worth noting that the kernelization itself can be done in time $O(m)$ since we process each vertex in constant time given that we know its degree. Thus, the kernelization does not affect the total asymptotic runtime of the algorithm.

\section{Conclusion and Open Problems}
We presented an FPT algorithm for the \crsv{k} problem parameterized by~$k$. Improving the runtime is needed for this algorithm to be useful in practice, as the constants are very large. 
Another natural direction is to look for an FPT algorithm for the other variant of the problem, that is, minimizing the number of splits, which was recently shown to be NP-hard~\cite{ahmed-split-vertices}. 
\begin{problem*}[Crossing Removal with $\boldsymbol{k}$ Splits  -- \crs{k}]
Let $G = (T \cup B, E)$ be a bipartite graph. Decide whether there is a planar 2-layer drawing of $G$ after applying at most $k$ splits to the vertices in $B$.
\end{problem*}
Is there an FPT algorithm for the \crs{k} problem parameterized  by $k$? It is not clear how to adjust the algorithm in Theorem~\ref{thm:crsvk-fpt} as it splits every vertex in~$B_\text{tr}$ as many times as its degree, and thus, the number of splits is not bounded by a function of the parameter $k$.

\bibliographystyle{splncs04}
\bibliography{refs}

\end{document}